\documentclass[letterpaper]{article} 
\usepackage{aaai2027}  
\usepackage[hyphens]{url}  
\usepackage{graphicx} 
\usepackage{natbib}  
\usepackage{caption} 
\usepackage{algorithm}
\usepackage{algorithmic}
\usepackage{amsmath}
\usepackage{amssymb}
\usepackage{amsthm}
\usepackage{tikz}
\usepackage{multirow}
\usepackage{array}
\usetikzlibrary{positioning,arrows.meta,fit,backgrounds,shapes.geometric,calc}
\newtheorem{theorem}{Theorem}
\newtheorem{lemma}{Lemma}
\newtheorem{corollary}{Corollary}
\newtheorem{proposition}{Proposition}
\newtheorem{definition}{Definition}

\usepackage{booktabs}
\usepackage{placeins}
\title{DP-MemView: A Memory Interface for Attribute-Level Transcript Privacy in Long-Term LLM Agents}
\author{
    Jong Wook Kim\textsuperscript{\rm 1},
    Byoungjae Min\textsuperscript{\rm 1},
    Kennedy Edemacu\textsuperscript{\rm 2},\\
    Yoonhyuk Choi\textsuperscript{\rm 3},
    Sae-Hong Cho\textsuperscript{\rm 4},
    Beakcheol Jang\textsuperscript{\rm 5}
}
\affiliations{
    \textsuperscript{\rm 1}Department of Computer Science, Sangmyung University, Seoul, Republic of Korea\\
    \textsuperscript{\rm 2}College of Staten Island, The City University of New York, New York, NY, USA\\
    \textsuperscript{\rm 3}Sookmyung Women's University, Seoul, Republic of Korea\\
    \textsuperscript{\rm 4}School of Computer Engineering, Hansung University, Seoul, Republic of Korea\\
    \textsuperscript{\rm 5}Graduate School of Information, Yonsei University, Seoul, Republic of Korea\\
    jkim@smu.ac.kr
}

\begin{document}
\nocopyright
\maketitle

\begin{abstract}
Long-term memory enables persistent personalization in LLM agents, but repeated memory-conditioned responses can cumulatively reveal protected attributes even when they are never stated explicitly. We formalize this threat as adaptive transcript privacy and introduce \textsc{DP-MemView}, a differentially private interface that privately selects public response-conditioning views and exposes those views---rather than raw memory---to the response LLM. Each private selection is charged to every protected attribute whose memory group intersects the read set. Per-attribute ledgers block any selection that would exceed its cap and return a fixed generic view instead. Under an explicit interface contract, we prove pure \(B_a\)-DP for the entire adaptive transcript. We also extend the result to stores that differ across multiple protected groups and bound how much observing the transcript can change an adversary's prior odds. We evaluate the online and preallocated modes with three response LLMs on a controlled adjacent-store benchmark and a public-corpus transfer track. Both modes keep transcript distinguishability near chance while preserving target-required personalization and overall response quality. Further diagnostics show that removing key safeguards causes mismatched output support, missing ledger charges, revealing side channels, or growing long-horizon leakage.
\end{abstract}


\section{Introduction}

Long-term memory now lets LLM agents store, retrieve, summarize, and update user-specific context across sessions, enabling persistent personalization and coherent interaction~\cite{park2023generative,packer2023memgpt,zhong2024memorybank,du2026memguide,huang2026mempal,tan2025prospect,dai2026memoryart,li2025hello,kang2025memoryos,xu2025amem,liu2025palace}.
This persistence, however, changes the privacy problem: a sensitive attribute need not appear in any single answer, but may become inferable from many benign-looking, memory-conditioned responses. This risk is becoming practically relevant as assistants connect persistent personal context to third-party applications: ChatGPT exposes applications through its Apps SDK~\cite{openai2025apps}, while Gemini's Personal Intelligence uses data from connected Google apps to personalize responses~\cite{googleGeminiConnectedApps}.

\begin{figure}[t]
\centering
\includegraphics[width=0.95 \columnwidth]{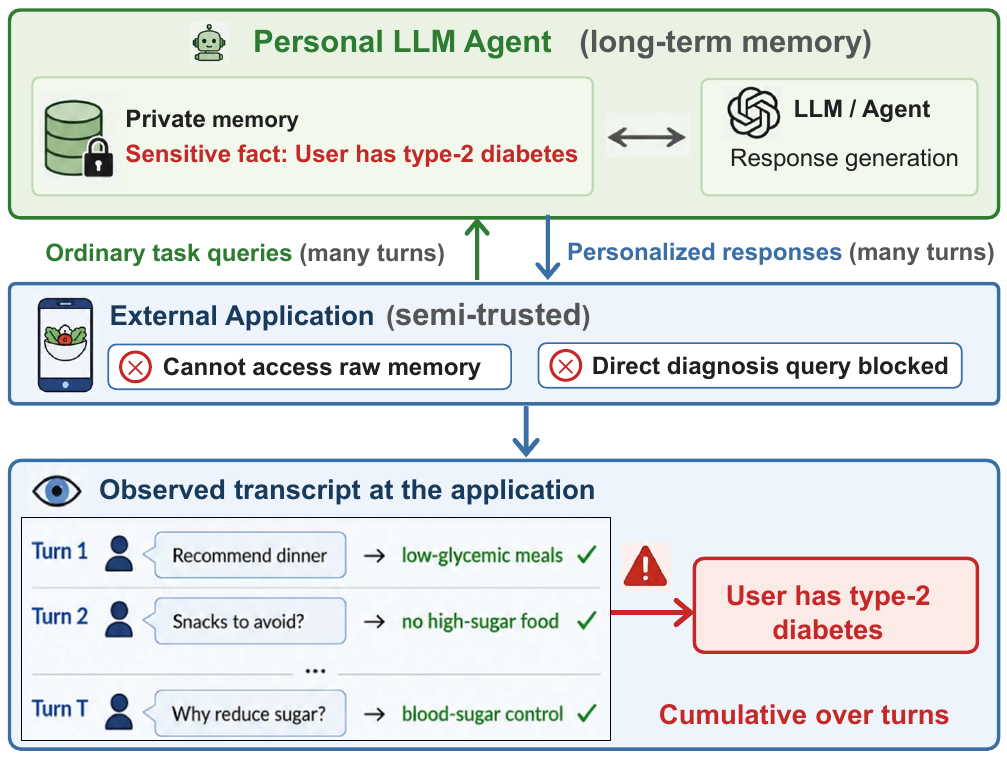}
\caption{Cumulative attribute leakage through benign personalization:
although direct memory access is blocked, repeated responses can jointly reveal the protected attribute.}
\label{fig:motiv}
\end{figure}

Recent work likewise shows that LLM agents may violate contextual privacy not only through explicit disclosure, but also through action trajectories, tool-mediated workflows, and task-conditioned information flows~\cite{mireshghallah2024confaide,shao2024privacylens,bagdasarian2024airgap,wang2025privacyaction}. In this setting, an external application need not read the raw memory store; it can submit task queries and observe the agent's
memory-conditioned responses over time. Thus, long-term memory privacy becomes an app-facing transcript problem rather than a single-output filtering problem.

Figure~\ref{fig:motiv} illustrates this risk. A meal-planning application uses the user's personal LLM agent to obtain personalized meal advice. The application may be blocked from directly asking whether the user has diabetes and may have no access to the raw memory store. Nevertheless, by issuing ordinary meal-planning queries, it can repeatedly observe blood-sugar-conscious recommendations produced from the agent's long-term memory. No individual response need state the diagnosis, yet the transcript can reveal it. Such leakage is a form of indirect disclosure: the sensitive attribute is not necessarily stated, but becomes inferable from repeated, task-relevant personalization signals.  

Existing work addresses related risks: extracting agent memories through
black-box or adaptive queries~\cite{wang2025mextra,lyu2026adam}, masking
sensitive spans~\cite{chen2026memprivacy,wu2026amp}, and protecting
retrieved documents or RAG outputs~\cite{choi2025ragmia,koga2024dprag,
tang2026dprag}. These defenses limit what is directly exposed, but they
do not provide a composable, prior-independent guarantee for privacy loss
accumulated across repeated memory-conditioned releases. We ask whether the memory
interface itself can make such releases compose into a single
differential privacy (DP) guarantee for the full adaptive transcript.

We introduce \textsc{DP-MemView}, a differentially private memory-view
interface. Instead of sending raw memories along the application-facing
generation path, it privately selects response-conditioning views from fixed public
vocabularies shared across adjacent memory stores and passes only those views
to the response-generating LLM. Views encode response-conditioning modes,
not user facts, so the application-facing response is post-processing of
the private view selections.

The view selector alone cannot ensure transcript privacy: failures arise from
store-dependent view support, incomplete charging of overlapping attributes,
raw-memory-dependent orchestration or side channels, and unbounded repetition
of locally private releases. \textsc{DP-MemView} enforces privacy at the
memory-interface boundary. Its interface contract fixes view support; keeps raw memory and internal
scores within the trusted interface; and requires every observable release
either to have the same conditional distribution across adjacent stores or
to be produced by a private selection mechanism, with its cost charged to
every affected attribute; restricts orchestration and controls to store-independent policy and released
state; enforces pathwise caps with a fixed generic fallback; and confines  LLM
responses to post-processing of public inputs and released state. Thus, the
interface---not the view selector alone---is the unit of privacy enforcement.


This paper makes three contributions.
\begin{itemize}

\item We formulate long-term-memory privacy as an adaptive
transcript-privacy problem, requiring protection for the complete sequence
of memory-conditioned interactions rather than for each response in
isolation.

\item We introduce \textsc{DP-MemView}, which privately selects public
slot-level views, charges every affected attribute, and enforces
per-attribute caps with a fixed generic fallback. By restricting routing
and observable controls to public inputs and released history, while
keeping raw memory and internal scores out of downstream generation, we
prove pure \(B_a\)-DP for the full adaptive transcript.

\item We evaluate online and preallocated budget modes with three response
LLMs on a controlled adjacent-store benchmark and a public-corpus transfer
track. Against memory-independent controls, raw-memory references, and
adapted state-of-the-art defenses, both modes keep transcript
distinguishability near chance while preserving personalization and
response quality. Ablations and diagnostics expose distinct accounting, release, channel,
and cumulative-budget failures.

\end{itemize}

\section{Related Work}

\noindent \textbf{Long-term memory and memory leakage in LLM agents.}
Recent work treats memory as a core agent capability, evaluating accurate
retrieval, test-time learning, long-range understanding, and selective
forgetting \cite{hu2025memoryagentbench}. Persistent memory, however, also
creates a new privacy surface. MEXTRA and ADAM show that stored agent
memories can be elicited through black-box prompting and adaptively selected
natural-language queries \cite{wang2025mextra,lyu2026adam}. These attacks
expose direct extraction of stored content, but do not bound the cumulative
influence of protected attributes encoded in memory on repeated,
task-conditioned outputs. DP-MemView instead defines an accounting boundary
at the memory interface so that such releases compose into a single DP
guarantee for the full adaptive transcript.

\noindent \textbf{Privacy-preserving memory and retrieval.}
Prior work limits direct exposure of sensitive content before storage,
retrieval, or transmission across trust boundaries. MemPrivacy replaces detected privacy spans with type-aware placeholders for cloud-side processing and restores their values locally \cite{chen2026memprivacy}, while the Agent-Memory Protocol stores placeholder-redacted memories, supplies only task-relevant redacted context to the model, and restores the original values locally after inference \cite{wu2026amp}. In RAG,
prior work either detects and hides documents targeted by
membership-inference queries or provides document-level DP for generated
answers \cite{choi2025ragmia,koga2024dprag,tang2026dprag}. Contextual-integrity
work evaluates whether disclosures fit the task and recipient, while
AirGapAgent addresses a closely related app-facing threat by minimizing
the user data made available to an untrusted third party
\cite{mireshghallah2024confaide,bagdasarian2024airgap}. These approaches
protect values, documents, or task-scoped disclosures, but do not define
an attribute-level accounting boundary for repeated memory-conditioned
releases to an adaptive application. DP-MemView places this boundary at
the memory interface, composing repeated releases into an attribute-level
DP guarantee for the full adaptive transcript.

\noindent \textbf{Leakage metrics and differential privacy.}
Quantitative information flow, min-entropy leakage, and \(g\)-leakage
quantify how observations change an adversary's ability to infer a secret
\cite{smith2009qif,alvim2012gleakage}. Such metrics characterize risk
under a specified prior, guessing objective, or gain function. DP instead bounds changes in output distributions between neighboring
inputs, independently of any particular adversary's prior or inference rule
\cite{dwork2014dp}. \textsc{DP-MemView} uses standard DP primitives—the
exponential mechanism (EM) and adaptive composition
\cite{mcsherry2007mechanism,dwork2014dp,kairouz2015composition,rogers2016privacy}.
Its contribution is an enforceable memory-interface contract that
composes adaptive memory-conditioned releases into an attribute-level
DP guarantee for the full transcript, not a new privacy primitive.

\section{Problem Definition}

We model long-term memory as an indexed store
\(M=(m_1,\ldots,m_n)\), where each \(m_i\) is the content at a fixed
memory position assigned to a policy-defined slot such as health, food, finance, work, or schedule. A protected attribute
is a user property, such as a health condition or financial hardship,
whose inference the protection policy seeks to limit. Adjacent stores
share the same memory layout: the same positions, identifiers, and
non-sensitive slot metadata.

At turn \(t\), an adaptive semi-trusted application issues a query
\(q_t\). The trusted memory interface selects a memory-view output
\(v_t\) for the response-generating LLM, which returns \(y_t\);
when multiple slots are active, \(v_t\) denotes the tuple of selected
slot-level views. Let \(c_t\) denote any observable control-flow signal,
such as a fallback or public budget-status signal. We define
\begin{equation}
\tau_T=\bigl((q_t,c_t,y_t)\bigr)_{t=1}^{T},
\qquad
\tilde{\tau}_T=\bigl((q_t,v_t,c_t,y_t)\bigr)_{t=1}^{T}.
\label{eq:transcripts}
\end{equation}
The application observes the ordinary transcript \(\tau_T\), whereas
the augmented transcript \(\tilde{\tau}_T\) additionally reveals the
selected views to a stronger analytical adversary. The adversary may
adapt queries to the released history, hold arbitrary auxiliary
information, and observe all signals in \(c_t\).

\noindent \textbf{Protected attributes and adjacency.}
A protected attribute may be implied by several memories rather than by
one explicit record. For each attribute \(a\in\mathcal{A}\), the
protection policy specifies an index set
\(I_a\subseteq\{1,\ldots,n\}\) containing all memory positions whose
contents may individually or jointly imply \(a\). We denote the
corresponding attribute-level memory group by
\(M_{I_a}=(m_i)_{i\in I_a}\). These groups may overlap because one memory
may provide evidence about multiple attributes.
Figure~\ref{fig:arch} makes this notation concrete: in the
running example, \(a_D\) denotes type-2 diabetes with
\(I_{a_D}=\{1,2,3\}\), whereas \(a_B\) denotes debt with
\(I_{a_B}=\{3,4,5\}\); thus, position \(m_3\) belongs to both
protected-attribute groups.
The guarantee is conditional on complete grouping: if a position that
may imply \(a\) is omitted from \(I_a\), changes to that position fall
outside \(\sim_a\) and hence outside the guarantee.

We use symmetric attribute-level group-replacement adjacency. Two memory
stores \(M\) and \(M'\) are adjacent with respect to attribute \(a\),
written \(M\sim_a M'\), if they share the same memory layout and agree at
every position outside \(I_a\), while positions within \(I_a\) may contain
different but slot-compatible contents under the same structural metadata.
The attribute-presence case compares a store in which \(M_{I_a}\) provides
evidence for \(a\) with a store in which \(M'_{I_a}\) is \(a\)-neutral.

\noindent \textbf{Privacy and trust boundary.}
Only the trusted memory interface may access raw memory. It performs all
memory-dependent scoring and view selection internally, while the
response-generating LLM receives \(q_t\), \(v_t\), and \(\tau_{<t}\), but
not \(M\). We model the LLM as a fixed, store-independent conditional channel:
\begin{equation}
P(y_t\mid q_t,v_t,\tau_{<t},M)
=
K_{\mathrm{LLM}}(y_t\mid q_t,v_t,\tau_{<t}).
\label{eq:fixed-llm}
\end{equation}
\(K_{\mathrm{LLM}}\) denotes the fixed conditional response
distribution mapping \((q_t,v_t,\tau_{<t})\) to \(y_t\), shared across
adjacent memory stores. Thus, conditioned on \((q_t,v_t,\tau_{<t})\), \(y_t\) is independent
of \(M\) and is post-processing of these released inputs.

\noindent \textbf{Pure-DP transcript privacy.}
For each \(a\in\mathcal{A}\), let \(B_a\ge0\) denote the
precommitted total pure-DP budget for the interaction. The interface
satisfies pure \(B_a\)-DP with respect to \(\sim_a\) if, for every
adjacent pair \(M\sim_a M'\), every adaptive query strategy, and every
transcript event \(\mathcal{O}\), the following condition holds:
\begin{equation}
\Pr[\tilde{\tau}_T\in\mathcal{O}\mid M]
\le
e^{B_a}
\Pr[\tilde{\tau}_T\in\mathcal{O}\mid M'].
\label{eq:transcript-dp}
\end{equation}
By symmetry, the same bound holds with \(M\) and \(M'\) exchanged.
Moreover, since \(\tau_T\) is a projection of \(\tilde{\tau}_T\), it
inherits the same guarantee by post-processing.

\begin{figure*}[t]
\centering
\includegraphics[width=2.02 \columnwidth]{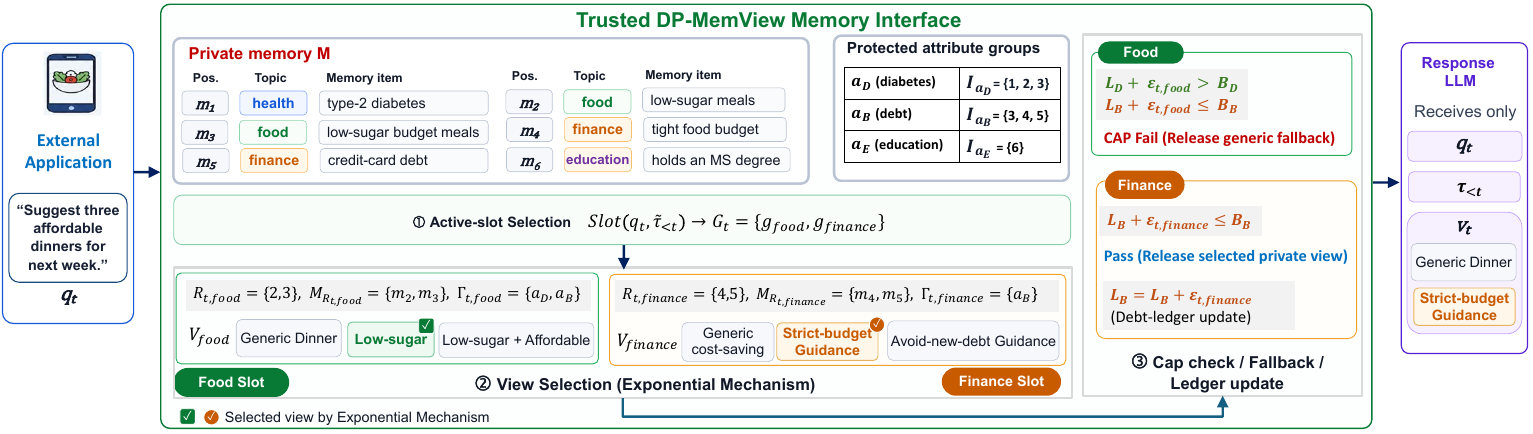}
\caption{Running example of \textsc{DP-MemView} for an affordable-dinner query. The interface activates the food and finance slots, maps each read set to the affected attributes, and applies the corresponding privacy-budget checks. The food slot releases the fixed generic fallback, while the finance slot releases the selected strict-budget view and charges the debt ledger. Only \(q_t\), the released view tuple \(v_t\), and \({\tau}_{<t}\) are passed to the response LLM; raw memory remains inside the trusted interface.}
\label{fig:arch}
\end{figure*}

\section{DP-MemView}

\textsc{DP-MemView} places the privacy boundary between private memory
and the response-generating LLM. Raw memory is used only inside the
trusted interface to select response-conditioning views; the LLM receives
only the selected views, and privacy cost is tracked by attribute-level
pure-DP ledgers. Figure~\ref{fig:arch} summarizes this interface flow.

\subsection{Slot-Level View Selection and Attribute Accounting}

Each memory position \(m_i\in M\) belongs to a policy-defined slot
\(g\in\mathcal{G}\) (e.g., health, food, finance, or work). For each slot,
\textsc{DP-MemView} fixes a public view vocabulary
\(\mathcal{V}_g\), shared across adjacent stores. The vocabulary includes
a fixed generic view \(v_g^{\mathrm{gen}}\), which provides
store-independent slot-level conditioning. A view
specifies how the LLM should answer rather than a raw user fact; for
instance, a health view may request blood-sugar-conscious advice without
asserting a diagnosis.

At turn \(t\), the router maps the public query and released history to
active slots,
\(G_t=\mathrm{Slot}(q_t,\tilde{\tau}_{<t})\).
For each active slot \(g\in G_t\), the read policy returns memory
positions \(\mathcal{R}_{t,g}\), with contents
\(M_{\mathcal{R}_{t,g}}=(m_i)_{i\in\mathcal{R}_{t,g}}\).
In our implementation, \(\mathcal{R}_{t,g}\) is determined by
\(q_t\), \(g\), the fixed memory layout, and
\(\tilde{\tau}_{<t}\), and does not depend on memory contents.
Content-dependent retrieval is not covered by the transcript-privacy
theorem; it would need to be adjacency-invariant or separately private,
with any retrieval cost charged to every affected attribute.

A read affects every attribute whose protected group intersects the read
set:
\begin{equation}
\Gamma^\star(\mathcal{R}_{t,g})
=
\{a\in\mathcal{A}:
\mathcal{R}_{t,g}\cap I_a\neq\emptyset\}.
\label{eq:affected-attributes}
\end{equation}
Intuitively, \(\Gamma^\star(R_{t,g})\) lists every protected attribute
whose memory group is touched by the read. These are the ledgers that
must be checked and, if the private selection is admitted, charged.
For example, in Figure~\ref{fig:arch}, \(R_{t,\mathrm{food}}=\{2,3\}\) intersects both \(I_{a_D}\) and \(I_{a_B}\), so \(\Gamma^\star(R_{t,\mathrm{food}})=\{a_D,a_B\}\).
The declared charge set \(\Gamma_{t,g}\subseteq\mathcal A\) must satisfy
\(\Gamma_{t,g}\supseteq\Gamma^\star(R_{t,g})\).
Over-charging is allowed, but no affected attribute may be omitted.

Then, the internal scorer assigns a score to each candidate view
using the public query and selected read-set memories:
\begin{equation}
u_{t,g}(q_t,v;M_{\mathcal{R}_{t,g}})\in[0,1].
\label{eq:readset-utility}
\end{equation}
Let \(\mathbf{u}_{t,g}(M) = \bigl(u_{t,g}(q_t,v;M_{\mathcal{R}_{t,g}})\bigr)_{v\in\mathcal{V}_g}\)
denote the internal score vector. For a charged selection, holding the
public inputs and released history fixed, we use a  bound
satisfying
\begin{equation}
\Delta u_{t,g}
\ge
\sup_{\substack{a\in\Gamma_{t,g},\; M\sim_a M'}}
\left\|
\mathbf{u}_{t,g}(M)-\mathbf{u}_{t,g}(M')
\right\|_\infty .
\label{eq:score-sensitivity}
\end{equation}
The score vector remains internal, since releasing it could reveal which
views the memory supports. If \(a\notin\Gamma_{t,g}\), then
\(\mathcal{R}_{t,g}\cap I_a=\emptyset\); hence, under
\(M\sim_a M'\), the read-set contents and score vectors are identical,
and the selection has zero privacy cost for \(a\)
(Lemma~\ref{lem:readset} in Appendix~\ref{app:proof-readset}).

Let \(\varepsilon_{t,g}\) be the local privacy parameter set by the
budget rule. For each charged selection, the interface samples a view
using the EM:
\begingroup
\small
\setlength{\abovedisplayskip}{4pt}
\setlength{\belowdisplayskip}{3pt}
\begin{equation}
\Pr[v_{t,g}=v\mid M]
=
\frac{
\exp\!\left(
\dfrac{\varepsilon_{t,g}\,
u_{t,g}(q_t,v;M_{\mathcal{R}_{t,g}})}
{2\Delta u_{t,g}}
\right)
}{
\displaystyle\sum_{v'\in\mathcal{V}_g}
\exp\!\left(
\dfrac{\varepsilon_{t,g}\,
u_{t,g}(q_t,v';M_{\mathcal{R}_{t,g}})}
{2\Delta u_{t,g}}
\right)
}.
\label{eq:exp-mech}
\end{equation}
\endgroup
\noindent
The EM favors useful views while bounding selection-probability
variation between adjacent stores.

For each protected attribute \(a\), let \(L_a\) denote the privacy cost
accumulated so far and \(B_a\) its precommitted cap. A charged selection
at \((t,g)\) is admitted only if
\(L_a+\varepsilon_{t,g}\le B_a\) for every
\(a\in\Gamma_{t,g}\). If any check fails, the interface returns the fixed
generic view at zero additional cost and leaves the ledgers unchanged.
Otherwise, it performs the private selection and updates
\(L_a\leftarrow L_a+\varepsilon_{t,g}\) for every
\(a\in\Gamma_{t,g}\). Thus, \(L_a\le B_a\) on every execution path.
Full pseudocode appears in Appendix~\ref{app:method-details}.

\subsection{End-to-End Transcript Privacy}
\label{sec:end-to-end-transcript}

The preceding subsection defines the slot-level release mechanism and
its pathwise budget rule. We now show that repeated adaptive
applications of that mechanism protect the complete transcript. The
key issue is state dependence: later slot operations may depend on
previously released views, but adjacent executions conditioned on the
same released prefix must use the same next operation and budget
branch. Raw memory may affect only the private view distribution.

Enumerate the active slots in
\(G_t=\mathrm{Slot}(q_t,\tilde{\tau}_{<t})\) according to the fixed
slot policy as \((g_{t,j})_{j=1}^{|G_t|}\). Before processing
\(g_{t,j}\), define $ h_{t,j} = \bigl(\tilde{\tau}_{<t},q_t,
(v_{t,g_{t,i}})_{i<j}\bigr)$. Thus, \(h_{t,j}\) contains the completed
history, the current query, and the views released by the preceding
slots, but neither raw memory nor internal scores.

\begin{definition}[Interface invariants for transcript privacy]
\label{def:memory-contract}
An interface with precommitted caps \(\{B_a\}\) satisfies the
invariants if every slot-level operation obeys the following.

\emph{Local release.}
\((\mathrm{C1})\) Every view vocabulary, including its fixed generic view, is shared
across adjacent stores. \((\mathrm{C2})\) Every read is charged to all affected attributes,
\(\Gamma_{t,g}\supseteq
\Gamma^\star(\mathcal{R}_{t,g})\). \((\mathrm{C3})\) Every released view is either adjacency-invariant or,
for an admitted charged selection, sampled by
Eq.~\eqref{eq:exp-mech}; raw memory and internal scores remain inside
the trusted interface.

\emph{Adaptive execution.}
\((\mathrm{C4})\) The active slots \(G_t\) and their processing order
are determined from \((q_t,\tilde{\tau}_{<t})\). Conditional on the same prefix \(h_{t,j}\), the read set, charge set,
local privacy parameter, and sensitivity bound for \(g_{t,j}\) are
identical across adjacent stores and use
neither raw memory nor unreleased scores. \((\mathrm{C5})\) The prospective cap check, generic fallback, and
ledger update are applied on every execution path, maintaining \(L_a\le B_a\) for every
\(a\in\mathcal A\).

\emph{Downstream release.}
\((\mathrm{C6})\) Conditioned on the public inputs, released history,
selected views, and the ledger state induced by that history, the
response LLM and all observable controls are independent of raw
memory.
\end{definition}

We refer to Conditions~\((\mathrm{C1})\)--\((\mathrm{C6})\)
collectively as the \textsc{DP-MemView} interface contract.
Conditions~\((\mathrm{C1})\)--\((\mathrm{C3})\) bound the privacy loss
of each slot-level release,
Conditions~\((\mathrm{C4})\)--\((\mathrm{C5})\) make these bounds
compose along the released execution path, and
Condition~\((\mathrm{C6})\) extends the guarantee to the
application-facing transcript.

\begin{theorem}[Adaptive transcript privacy of \textsc{DP-MemView}]
\label{thm:transcript}
If the interface satisfies Definition~\ref{def:memory-contract}, then,
for every protected attribute \(a\), every adaptive query strategy, and
every adjacent pair \(M\sim_a M'\), the augmented transcript
\(\tilde{\tau}_T\) satisfies pure \(B_a\)-DP with respect to
\(\sim_a\), as defined in Eq.~\eqref{eq:transcript-dp}. The ordinary
transcript \(\tau_T\) satisfies the same guarantee by post-processing.
\end{theorem}

Equation~\eqref{eq:exp-mech} protects a single slot-level view
selection, whereas Theorem~\ref{thm:transcript} protects the complete
adaptive interaction. If every slot-level operation satisfies
Conditions~\((\mathrm{C1})\)--\((\mathrm{C6})\), then the queries,
released views, observable controls, and LLM responses accumulated over
multiple turns jointly satisfy \(B_a\)-DP for each protected attribute
\(a\). The guarantee holds even when the application chooses each query
adaptively from the observed history. Because it is established for the
augmented transcript, which additionally reveals the selected views, the
application-facing transcript inherits the same bound. The full proof is
in Appendix~\ref{app:proof-transcript}.

Because the prospective check maintains \(L_a\le B_a\) on every
execution path, every finite transcript prefix satisfies
\(B_a\)-DP. If a later selection would exceed the cap, the interface
returns \(v_g^{\mathrm{gen}}\) at zero additional privacy cost, so the
cumulative privacy loss of the transcript remains bounded by \(B_a\)
and the transcript-level \(B_a\)-DP guarantee continues to hold.

Theorem~\ref{thm:transcript} protects two stores that differ within one
protected group. We next consider two stores that may differ across
several groups. Let \(M=(m_1,\ldots,m_n)\) and \(M'=(m'_1,\ldots,m'_n)\) be two
same-scaffold stores, and let
\(D=\{i:m_i\neq m'_i\}\) denote the memory positions whose contents
differ between them. A set of protected attributes
\(\mathcal C\subseteq\mathcal A\) covers these differences if every
position in \(D\) belongs to at least one protected group \(I_a\) with
\(a\in\mathcal C\), that is,
\(D\subseteq\bigcup_{a\in\mathcal C}I_a\).
Chaining Theorem~\ref{thm:transcript} over the attributes in
\(\mathcal C\) gives the composite privacy cost
\(\sum_{a\in\mathcal C}B_a\).

\begin{corollary}[Composite-store transcript privacy]
\label{cor:attribute-cover}
For same-scaffold stores \(M\) and \(M'\), define the minimum cover cost
\begin{equation}
\kappa(M,M')
=
\min_{\substack{\mathcal C\subseteq\mathcal A, \;\; D\subseteq\bigcup_{a\in\mathcal C}I_a}}
\sum_{a\in\mathcal C}B_a,
\label{eq:attribute-cover}
\end{equation}
with \(\kappa(M,M')=\infty\) if no protected-group cover exists.
If \(\kappa(M,M')<\infty\) and the interface satisfies
Definition~\ref{def:memory-contract}, then, for every
adaptive query strategy and transcript event  \(\mathcal O\),
\begin{equation}
\Pr[\tilde{\tau}_T\in\mathcal O\mid M]
\le
e^{\kappa(M,M')}
\Pr[\tilde{\tau}_T\in\mathcal O\mid M'].
\label{eq:composite-transcript}
\end{equation}
The same bound holds for the ordinary transcript.
\end{corollary}

Thus, Corollary~\ref{cor:attribute-cover} extends the per-attribute
guarantee of Theorem~\ref{thm:transcript} to stores differing across
multiple protected groups, using the minimum total budget of groups
covering all differing memory positions. This completes the
transcript-level privacy guarantee for composite store differences.
The proof is in Appendix~\ref{app:attribute-cover}.

\noindent\textbf{Inference interpretation.}
The transcript-level guarantee also bounds attribute inference: for any
prior odds between adjacent stores \(M\) and \(M'\), observing a
transcript event with positive probability under both hypotheses can
multiply or divide those odds by at most \(e^{B_a}\).
This bound is independent of the prior; see
Appendix~\ref{app:posterior-odds}.

\noindent\textbf{Why the full contract is needed.}
The transcript guarantees above do not follow from the private selector
alone. Conditions~\((\mathrm{C1})\)--\((\mathrm{C3})\) secure each
slot-level release, Conditions~\((\mathrm{C4})\)--\((\mathrm{C5})\)
make its privacy cost compose along the realized execution path, and
Condition~\((\mathrm{C6})\) prevents downstream channels from bypassing
the interface. Within this architecture and absent additional
safeguards, violating common support, complete charging, control of
memory-dependent channels, or pathwise cap enforcement can break
transcript-level \(B_a\)-DP and therefore the inference bound above.
Appendix~\ref{app:contract-failures} gives isolating counterexamples.

\section{Experiments}
\label{sec:experiments}

\subsection{Research Questions}

\begin{itemize}

\item \emph{(RQ1) Privacy--utility.}
Across repeated  interactions, does the DP-MemView reduce
observed transcript-level attribute distinguishability while preserving
supported personalization and limiting unsupported personalization,
relative to memory-independent controls, raw-memory references, and
adapted state-of-the-art privacy defenses?

\item \emph{(RQ2) Interface safeguards.}
Is private view selection alone sufficient, or do violations of the
surrounding interface contract produce distinct privacy, accounting, and
personalization failures?

\end{itemize}

\begin{table}[t]
\centering
\scriptsize
\setlength{\tabcolsep}{3pt}
\renewcommand{\arraystretch}{1.0}
\caption{Compared memory interfaces.}
\label{tab:compared-interfaces}
\begin{tabular}{@{}p{0.30\columnwidth}p{0.67\columnwidth}@{}}
\toprule
\textbf{Method} & \textbf{Response-LLM input; protection} \\
\midrule

\multicolumn{2}{@{}l}{\textbf{Memory-independent}}\\[-1mm]
GenericOnly
& Fixed generic view; memory-independent \\

\addlinespace[1pt]
\multicolumn{2}{@{}l}{\textbf{Memory text}}\\[-1mm]
RawReadSet
& Raw read-set text; non-private reference \\
TypedMask
& Item-masked read set; MemPrivacy adaptation \\
TaskMin
& Fixed task-family subset; AirGapAgent adaptation \\
OutputFilter
& Raw read-set text; post-generation deletion filter \\

\addlinespace[1pt]
\multicolumn{2}{@{}l}{\textbf{Proposed modes}}\\[-1mm]
\textbf{DP-MemView (on)}
& EM-sampled public view or generic fallback; Online allocation + ledger/cap \\

\textbf{DP-MemView (pre)}
& EM-sampled public view or generic fallback; Planned allocation + ledger/cap \\
\bottomrule
\end{tabular}
\end{table}

\begin{table*}[t]
\centering
\scriptsize
\setlength{\tabcolsep}{4.0pt}
\renewcommand{\arraystretch}{0.76}
\caption{RQ1 results at \(T=16\) and \(B_a=2\), with \(K=3\) for
DP-MemView (on).
\textbf{Privacy:} AUC closer to \(0.5\) and TPR@5 closer to \(0.05\)
indicate lower distinguishability.
\textbf{Utility:} Higher \(\mathrm{tRec}\) and \(U\) are better, whereas
lower \(\mathrm{Unsup}\) is better. \(\dagger\) denotes identical conditioning under \(M\) and \(M'\),
yielding chance-level distinguishability by construction.}
\label{tab:rq1-main}
\begin{tabular}{@{}l*{3}{rrrrr}@{}}
\toprule
\textbf{Method}
& \multicolumn{5}{c}{\textbf{Qwen2.5-7B}}
& \multicolumn{5}{c}{\textbf{Llama-3.1-8B}}
& \multicolumn{5}{c}{\textbf{Gemma-2-9B}}\\
\cmidrule(lr){2-6}\cmidrule(lr){7-11}\cmidrule(l){12-16}
& AUC & TPR@5\(\downarrow\) & tRec\(\uparrow\) & U\(\uparrow\) & Unsup\(\downarrow\)
& AUC & TPR@5\(\downarrow\) & tRec\(\uparrow\) & U\(\uparrow\) & Unsup\(\downarrow\)
& AUC & TPR@5\(\downarrow\) & tRec\(\uparrow\) & U\(\uparrow\) & Unsup\(\downarrow\)\\
\midrule
\multicolumn{16}{@{}l}{\textbf{PairedMem (synthetic)}}\\
\quad GenericOnly
& \(0.500^{\dagger}\) & \(0.050^{\dagger}\) & 0.426 & 0.877 & 0.044
& \(0.500^{\dagger}\) & \(0.050^{\dagger}\) & 0.429 & 0.854 & 0.065
& \(0.500^{\dagger}\) & \(0.050^{\dagger}\) & 0.378 & 0.882 & 0.025\\
\quad RawReadSet
& 0.842 & 0.547 & 0.521 & 0.775 & 0.136
& 0.942 & 0.688 & 0.525 & 0.665 & 0.271
& 0.925 & 0.609 & 0.490 & 0.700 & 0.237\\
\quad TypedMask
& \(0.500^{\dagger}\) & \(0.050^{\dagger}\) & 0.386 & 0.846 & 0.109
& \(0.500^{\dagger}\) & \(0.050^{\dagger}\) & 0.296 & 0.678 & 0.233
& \(0.500^{\dagger}\) & \(0.050^{\dagger}\) & 0.309 & 0.750 & 0.195\\
\quad TaskMin
& 0.832 & 0.406 & 0.524 & 0.772 & 0.129
& 0.963 & 0.719 & 0.532 & 0.650 & 0.229
& 0.919 & 0.641 & 0.468 & 0.729 & 0.202\\
\quad OutputFilter
& 0.793 & 0.438 & 0.517 & 0.786 & 0.136
& 0.885 & 0.562 & 0.509 & 0.666 & 0.266
& 0.866 & 0.469 & 0.475 & 0.712 & 0.238\\
\addlinespace[1pt]
\quad \textbf{DP-MemView (on)}
& 0.490 & 0.016 & 0.506 & 0.877 & 0.119
& 0.535 & 0.031 & 0.513 & 0.861 & 0.135
& 0.490 & 0.062 & 0.508 & 0.878 & 0.115\\
\quad \textbf{DP-MemView (pre)}
& 0.476 & 0.062 & 0.559 & 0.877 & 0.175
& 0.474 & 0.016 & 0.559 & 0.857 & 0.174
& 0.464 & 0.016 & 0.571 & 0.870 & 0.170\\
\midrule
\multicolumn{16}{@{}l}{\textbf{Public-corpus transfer}}\\
\quad GenericOnly
& \(0.500^{\dagger}\) & \(0.050^{\dagger}\) & 0.427 & 0.884 & 0.045
& \(0.500^{\dagger}\) & \(0.050^{\dagger}\) & 0.372 & 0.858 & 0.043
& \(0.500^{\dagger}\) & \(0.050^{\dagger}\) & 0.370 & 0.880 & 0.020\\
\quad RawReadSet
& 0.730 & 0.350 & 0.470 & 0.727 & 0.190
& 0.723 & 0.425 & 0.351 & 0.555 & 0.339
& 0.680 & 0.275 & 0.393 & 0.730 & 0.208\\
\quad TypedMask
& \(0.500^{\dagger}\) & \(0.050^{\dagger}\) & 0.369 & 0.815 & 0.139
& \(0.500^{\dagger}\) & \(0.050^{\dagger}\) & 0.289 & 0.603 & 0.270
& \(0.500^{\dagger}\) & \(0.050^{\dagger}\) & 0.304 & 0.782 & 0.131\\
\quad TaskMin
& 0.751 & 0.450 & 0.449 & 0.752 & 0.163
& 0.742 & 0.475 & 0.355 & 0.540 & 0.352
& 0.630 & 0.250 & 0.371 & 0.724 & 0.192\\
\quad OutputFilter
& 0.701 & 0.300 & 0.448 & 0.728 & 0.175
& 0.669 & 0.288 & 0.350 & 0.550 & 0.316
& 0.647 & 0.237 & 0.370 & 0.722 & 0.201\\
\addlinespace[1pt]
\quad \textbf{DP-MemView (on)}
& 0.562 & 0.113 & 0.513 & 0.883 & 0.105
& 0.563 & 0.100 & 0.503 & 0.862 & 0.114
& 0.535 & 0.100 & 0.507 & 0.878 & 0.089\\
\quad \textbf{DP-MemView (pre)}
& 0.523 & 0.062 & 0.566 & 0.878 & 0.172
& 0.490 & 0.113 & 0.565 & 0.856 & 0.179
& 0.486 & 0.088 & 0.573 & 0.873 & 0.160\\
\bottomrule
\end{tabular}
\end{table*}

\subsection{Experimental Setup}
We briefly summarize the experimental setup below; full details are
provided in Appendix~\ref{app:experiments}.

\paragraph{Evaluation data.}
No public long-term-memory corpus provides matched stores with the same
scaffold and non-target memories while differing only within \(I_a\).
We additionally hold the query trajectory fixed across each pair.
Thus, we use two complementary paired tracks.

\noindent\textbf{Synthetic controlled track (PairedMem).}
PairedMem provides exact counterfactual control through 320 pairs spanning
four protected attributes. Paired stores share their structure and
non-target content and differ only at 4--6 slot-compatible positions in
\(I_a\); each pair carries a 32-turn trajectory (10{,}240 gold-annotated
turns in total). Structure and gold annotations are fixed
programmatically, while language models generate only the surface wording.

\noindent\textbf{Public-corpus transfer track.}
To test whether the findings depend on the wording and templates of our
synthetic PairedMem construction, we construct 80 pairs whose memory
text is adapted from publicly released corpora---LoCoMo persona and
event histories~\cite{maharana2024locomo}, medical-question text from
the Medical Question Answering Datasets
collection~\cite{malikehmedicalqa}, and empathetic-dialogue
evidence~\cite{rashkin2019towards}---while the PairedMem skeleton, queries, policy, scorer, and
gold rules remain fixed.

\paragraph{Compared interfaces.}
All methods use the same query trajectory and response model and differ
only in the memory conditioning supplied at each turn and its protection
(Table~\ref{tab:compared-interfaces}). TypedMask and TaskMin are
setting-specific adaptations of MemPrivacy~\cite{chen2026memprivacy} and AirGapAgent~\cite{bagdasarian2024airgap}, respectively, while OutputFilter is a generic post-generation control. We evaluate DP-MemView
under two budget policies: the online mode admits fixed local charges
through a running ledger, whereas the preallocated mode distributes each
budget across known planned charges.

\paragraph{Query and budget protocol.}
Each pair is evaluated with 32 benign queries arranged as two
nonoverlapping 16-turn blocks. Each block contains 8 target-relevant,
4 same-domain attribute-neutral, and 4 distractor queries;
protected-attribute names and direct synonyms are excluded. The same query
sequence is used for both stores and all interfaces. Main results use \(T=16\).
We evaluate per-attribute caps \(B_a\in\{1,2,4,8\}\).
DP-MemView (on) uses budget granularities
\(K\in\{2,3,4\}\) and sets
\(\varepsilon_{t,g}=\min_{a\in\Gamma_{t,g}} B_a/K\).
DP-MemView (pre) instead distributes each budget over the planned
number \(C_a\) of charged selections for attribute \(a\), setting
\(\varepsilon_{t,g}
=\min_{a\in\Gamma_{t,g}} B_a/C_a\). Both modes use the same prospective cap check and generic fallback.

\paragraph{Models, metrics, and statistics.}
Responses are generated by Qwen2.5-7B-Instruct,
Llama-3.1-8B-Instruct, and Gemma-2-9B-it.
For each response model, privacy is measured by a frozen sentence encoder
and separate calibrated logistic-regression classifiers for each
attribute and prefix length, trained on PairedMem training transcripts,
calibrated on validation transcripts, and evaluated on held-out test
transcripts. The resulting auditors are then applied
without refitting to the corresponding public-corpus transfer transcripts. We report macro-AUC
and TPR at 5\% FPR.  

A blinded Claude Sonnet 4.6 judge evaluates response utility using a
frozen rubric. It scores each 16-turn transcript against a side-specific
checklist of required and unsupported response behaviors, without seeing
the method or store label. The judge assigns 0--4 scores to relevance, correctness, actionability,
and personalization. The normalized average of these four scores defines
\(U\). \(\mathrm{tRec}\) averages the normalized scores for
target-required behaviors on \(M\), whereas
\(\mathrm{Unsup}\) averages those for unsupported target-specific
behaviors on \(M'\). Thus, \(U\) reflects overall response quality,
\(\mathrm{tRec}\) reflects how well supported target-specific
information is used for personalization, and
\(\mathrm{Unsup}\) reflects the extent of unsupported
target-specific personalization.

\subsection{RQ1: Privacy--Utility}

Table~\ref{tab:rq1-main} compares DP-MemView with memory-independent
and raw-memory references and with setting-specific adaptations of
state-of-the-art privacy defenses. GenericOnly excludes private memory,
whereas RawReadSet directly supplies the read-set memory text to the
response model. TypedMask and TaskMin respectively evaluate masking and
task-based minimization, while OutputFilter is a generic post-generation
control.

\noindent\textbf{Privacy close to the memory-independent reference.}
GenericOnly yields chance-level distinguishability because its
conditioning is independent of the store. TypedMask also yields
chance-level distinguishability because it replaces the protected items
that differ between \(M\) and \(M'\) with identical placeholders,
removing the target-dependent signal. In contrast, RawReadSet, TaskMin,
and OutputFilter remain distinguishable across all models and both
tracks. Supplying read-set text, restricting memory to a task-specific
subset, or filtering the generated response does not prevent protected
attributes from influencing the transcript. Both DP-MemView modes remain
much closer to the memory-independent reference than these memory-text
methods while retaining memory-dependent conditioning. TPR@5 shows the
same separation at a 5\% false-positive rate.

\noindent\textbf{Useful personalization without raw-memory exposure.}
GenericOnly attains high overall response quality but low
\(\mathrm{tRec}\) because its fixed conditioning cannot provide
memory-grounded personalization. TypedMask also has low
\(\mathrm{tRec}\) because its placeholders remove the target-specific
information needed for personalization. Providing more memory text does
not produce higher overall response quality: across all model--track
combinations, both DP-MemView modes achieve higher \(U\) than RawReadSet,
TaskMin, and OutputFilter. This suggests that task-specific views guide
the response model more effectively than raw or filtered memory text.

DP-MemView (pre) achieves the highest \(\mathrm{tRec}\) in every
model--track combination while remaining close to chance in
distinguishability. DP-MemView (on) achieves comparable \(U\) but lower
\(\mathrm{Unsup}\) than the preallocated mode, RawReadSet, TaskMin, and
OutputFilter. The two modes therefore serve different settings:
preallocation sustains more target-required personalization when planned
charge counts are known, whereas online accounting limits unsupported
personalization without future charge information.

\noindent\textbf{Model and corpus robustness.}
The same privacy--utility separation holds across Qwen, Llama, and Gemma
and transfers to corpus-sourced memory text without auditor refitting.

\noindent\textbf{Budget sensitivity.}
Table~\ref{tab:budget-sensitivity} separates the roles of \(B_a\) and
\(K\). At fixed \(K=3\), AUC, \(\mathrm{tRec}\), and
\(\mathrm{Unsup}\) change little for \(B_a\le4\), whereas \(B_a=8\)
increases distinguishability. At fixed \(B_a=2\), increasing \(K\)
permits more charged releases at smaller \(\varepsilon_{t,g}\), increasing
both \(\mathrm{tRec}\) and \(\mathrm{Unsup}\). Thus, \(B_a\) scales
the local EM parameter, whereas \(K\) determines how many equal charges
each attribute budget supports.

\begin{table}[t]
\centering
\scriptsize
\setlength{\tabcolsep}{9.5pt}
\renewcommand{\arraystretch}{0.8}
\caption{Budget sensitivity of DP-MemView (on) on
Qwen2.5-7B and PairedMem at \(T=16\), with
\(\varepsilon_{t,g}=B_a/K\).}
\label{tab:budget-sensitivity}
\begin{tabular}{@{}llrrrr@{}}
\toprule
& \textbf{Setting}
& \(\boldsymbol{\varepsilon_{t,g}}\)
& \textbf{AUC}
& \textbf{tRec}\(\uparrow\)
& \textbf{Unsup}\(\downarrow\) \\
\midrule
\multirow{4}{*}{\shortstack[l]{\textbf{Cap}\\\textbf{magnitude}\\\((K=3)\)}}
& \(B_a=1\) & 0.333 & 0.541 & 0.508 & 0.115  \\
& \(B_a=2\) & 0.667 & 0.490 & 0.506 & 0.119  \\
& \(B_a=4\) & 1.333 & 0.508 & 0.513 & 0.113  \\
& \(B_a=8\) & 2.667 & 0.608 & 0.534 & 0.105  \\
\midrule
\multirow{3}{*}{\shortstack[l]{\textbf{Budget}\\\textbf{granularity}\\\((B_a=2)\)}}
& \(K=2\) & 1.000 & 0.510 & 0.487 & 0.088  \\
& \(K=3\) & 0.667 & 0.490 & 0.506 & 0.119  \\
& \(K=4\) & 0.500 & 0.525 & 0.527 & 0.129  \\
\bottomrule
\end{tabular}
\end{table}

\noindent\textbf{RQ1 summary.}
Across all three response models and both data tracks, both DP-MemView
modes remain much less distinguishable than RawReadSet, TaskMin, and
OutputFilter and achieve higher \(U\) than these methods. Relative to
GenericOnly and TypedMask, both modes recover more target-required
personalization. The preallocated mode provides the highest
\(\mathrm{tRec}\) when planned charge counts are known, whereas the
online mode yields lower \(\mathrm{Unsup}\) without future charge
information.

\subsection{RQ2: Interface Safeguards}

Table~\ref{tab:contract-ablations} isolates C2, C3, and C5 while
holding the remaining interface components fixed. We report the reconstructed target-attribute pure-DP bound and
response-level measurements at \(T=16\).

\emph{IncompleteCharge} violates C2 by charging only one of the
affected attributes when a read intersects multiple protected groups.
The table reconstructs the target attribute's cumulative privacy cost
by including the charges omitted from its ledger.
\emph{CappedArgmax} violates C3 by replacing EM sampling with
deterministic selection of the highest-scoring view. The ledger update and
generic fallback are left unchanged so that only the selection rule differs
from DP-MemView. If adjacent stores have different highest-scoring views,
one selected view can have probability 1 under one store and 0 under the
other; hence no finite pure-DP bound exists, and we report
\(\varepsilon=\infty\). \emph{NoCap} violates C5: it retains valid EM releases and complete
charging but removes the prospective cap check and generic fallback,
allowing cumulative charges to exceed \(B_a\).

\begin{table}[t]
\centering
\scriptsize
\setlength{\tabcolsep}{3.2pt}
\renewcommand{\arraystretch}{0.72}
\caption{End-to-end safeguard ablations on Qwen2.5-7B and PairedMem
at \(T=16\), with \(B_a=2\) and \(K=3\).}
\label{tab:contract-ablations}
\begin{tabular}{@{}lccrrrr@{}}
\toprule
\textbf{Variant}
& \textbf{Violated}
& \textbf{DP bound}
& \textbf{AUC}
& \(\mathbf{tRec}\uparrow\)
& \(\mathbf{U}\uparrow\)
& \(\mathbf{Unsup}\downarrow\) \\
\midrule
\textbf{DP-MemView (on)}
& --- & \(2.0\le B_a\)
& 0.490 & 0.506 & 0.877 & 0.119 \\

IncompleteCharge
& C2 & \(4.0>B_a\)
& 0.554 & 0.565 & 0.879 & 0.160 \\

CappedArgmax
& C3 & \(\infty\)
& 0.754 & 0.597 & 0.876 & 0.047 \\

NoCap
& C5 & \(5.33>B_a\)
& 0.599 & 0.583 & 0.876 & 0.161 \\
\bottomrule
\end{tabular}
\end{table}

\noindent\textbf{Personalization gains outside the \(B_a\)-DP guarantee.}
All three ablations increase tRec over DP-MemView, but none preserves its
\(B_a=2\) transcript guarantee. IncompleteCharge retains valid EM
releases but omits affected target charges: its ledger reports \(2.0\),
whereas the reconstructed charge is \(4.0\). CappedArgmax retains the
operational ledger and fallback only as controls, but deterministic
highest-score selection has no finite per-release DP bound. NoCap uses valid, fully charged EM releases, but their cumulative charge already reaches \(5.33\) at \(T=16\).

These violations also increase transcript distinguishability: AUC rises
from 0.490 for DP-MemView to 0.554, 0.754, and 0.599. The charge analysis
shows why the \(B_a\)-DP certificate fails, while AUC shows the resulting
empirical leakage. Because \(U\) changes by at most 0.002, the apparent
gain is mainly higher personalization obtained after invalidating or
exceeding the privacy guarantee. CappedArgmax's low Unsup is not evidence
of privacy; it accompanies the highest AUC.

\noindent\textbf{Direct support and side-channel failures.}
C1 is categorical: if a view is available under only one of two adjacent
stores, releasing that view has positive probability on one side and zero
on the other, so no finite pure-DP bound exists. For the C3/C6
channel-isolation case, exposing the top internal score yields a
score-channel AUC of 0.976, although the compliant view-selection and
response paths remain unchanged. Because our evaluation activates one
slot per turn, it does not isolate the slot-order subcase of C4; this case
is covered by the interface contract and its counterexample.

\noindent\textbf{RQ2 summary.}
Removing the safeguards causes failures at different layers: incomplete
charging makes the attribute ledger incorrect, deterministic selection
invalidates the local privacy bound, and removing the cap allows valid
local charges to accumulate beyond \(B_a\). Common support, safe
orchestration, and downstream isolation prevent one-sided outputs and
direct side channels. Thus, private view selection alone is insufficient;
the interface contract is required for transcript-level privacy.

\section{Conclusion}
We formulated long-term-memory privacy as an adaptive transcript problem
and introduced \textsc{DP-MemView}, which keeps raw memory inside a trusted
interface, passes only public views to the response LLM, and enforces
attribute-level pathwise caps. Under the interface contract, the full
adaptive transcript satisfies pure \(B_a\)-DP for each protected
attribute. Across three response LLMs and two evaluation tracks,
\textsc{DP-MemView} kept transcript distinguishability near chance while
retaining target-required personalization and response quality. Safeguard
ablations and diagnostics further showed that private selection alone is
insufficient and that the complete interface contract is necessary.
\FloatBarrier
\bibliography{aaai2027}
\clearpage

\appendix
\setcounter{secnumdepth}{2}

\clearpage
\onecolumn
\appendix

\section{Supplementary Method Details}
\label{app:method-details}

At the start of an interaction, \(L_a=0\) for every
\(a\in\mathcal A\). In our implementation, \(G_t\) contains only
memory-conditioned slots whose fixed read sets are nonempty and
intersect at least one protected group; hence, for every
\(g\in G_t\),
\(\mathcal R_{t,g}\neq\emptyset\) and
\(\Gamma^\star(\mathcal R_{t,g})\neq\emptyset\).
Algorithm~\ref{alg:dpmemview} gives the complete one-turn procedure
used in our implementation. Our implementation uses the minimal sound
charge set
\(\Gamma_{t,g}=\Gamma^\star(\mathcal R_{t,g})\), while
Definition~\ref{def:memory-contract} also permits conservative
supersets. All orchestration rules use only public inputs,
policy-fixed structure, and previously released outputs, as required
by Condition~(C4); cap checks and ledger updates follow
Condition~(C5). Below, \(\tau_{<t}\) denotes the ordinary-transcript
projection of \(\tilde{\tau}_{<t}\).

\begin{algorithm}[H]
\caption{\textsc{DP-MemView} at turn \(t\)}
\label{alg:dpmemview}
\small
\begin{algorithmic}[1]

\REQUIRE Query \(q_t\), memory \(M\), released history
\(\tilde{\tau}_{<t}\), ledgers and caps
\(\{(L_a,B_a)\}_{a\in\mathcal A}\), fixed policies, scorer,
sensitivity-bound rule, and view vocabularies
\(\{\mathcal V_g\}\)

\STATE
\(G_t\leftarrow\mathrm{Slot}(q_t,\tilde{\tau}_{<t})\)

\STATE enumerate \(G_t\) in the policy-fixed order as
\((g_{t,j})_{j=1}^{|G_t|}\)

\FOR{\(j=1,\ldots,|G_t|\)}

    \STATE \(g\leftarrow g_{t,j}\)

    \STATE
    \(h_{t,j}\leftarrow
    \bigl(\tilde{\tau}_{<t},q_t,
    (v_{t,g_{t,i}})_{i<j}\bigr)\)

    \STATE
    \(\mathcal R_{t,g}\leftarrow
    \mathrm{ReadPolicy}(q_t,g,h_{t,j})\)

        \STATE
        \(\Gamma_{t,g}\leftarrow
        \Gamma^\star(\mathcal R_{t,g})\)

        \STATE
        \(\varepsilon_{t,g}\leftarrow
        \mathrm{BudgetRule}
        \bigl(q_t,g,h_{t,j},\Gamma_{t,g},
        \{(L_a,B_a)\}_{a\in\Gamma_{t,g}}\bigr)\)

        \IF{\(\exists a\in\Gamma_{t,g}:
        L_a+\varepsilon_{t,g}>B_a\)}

            \STATE
            \(v_{t,g}\leftarrow v_g^{\mathrm{gen}}\)

        \ELSE

            \STATE compute the internal score vector
            \(\mathbf u_{t,g}(M)\)
            
            \STATE obtain the policy-fixed valid sensitivity bound
            \(\Delta u_{t,g}>0\)
            
            \STATE sample \(v_{t,g}\) according to
            Eq.~\eqref{eq:exp-mech}

            \STATE
            \(L_a\leftarrow L_a+\varepsilon_{t,g}\),
            \(\forall a\in\Gamma_{t,g}\)
       
    \ENDIF
\ENDFOR

\STATE
\(v_t\leftarrow
(v_{t,g_{t,j}})_{j=1}^{|G_t|}\)

\STATE
\(c_t\leftarrow
\mathrm{ObservableControl}
\bigl(q_t,v_t,\tilde{\tau}_{<t},
\{(L_a,B_a)\}_{a\in\mathcal A}\bigr)\)

\STATE
\(y_t\sim
K_{\mathrm{LLM}}(\cdot\mid q_t,v_t,\tau_{<t})\)

\STATE append \((q_t,v_t,c_t,y_t)\) to the augmented history

\RETURN \((c_t,y_t)\) to the application
\end{algorithmic}
\end{algorithm}

\section{Proofs}
\label{app:proofs}

\subsection{Zero Privacy Cost for an Uncharged Attribute}
\label{app:proof-readset}

\begin{lemma}[Zero cost for an uncharged attribute]
\label{lem:readset}
Assume Conditions~\((\mathrm{C1})\)--\((\mathrm{C4})\).
Assume that the scorer uses only the read-set memories, as specified
in Eq.~\eqref{eq:readset-utility}, and that
\(\Delta u_{t,g}\) is fixed by the interface policy.

Fix the same prefix \(h\) under adjacent stores \(M\sim_a M'\), and
consider an admitted charged selection at \((t,g)\). If
\(a\notin\Gamma_{t,g}\), then, for every
\(v\in\mathcal V_g\),
\[
\Pr[v_{t,g}=v\mid h,M]
=
\Pr[v_{t,g}=v\mid h,M'].
\]
Thus, the selection adds zero privacy loss for attribute \(a\).
\end{lemma}

\begin{proof}
Under the same prefix \(h\), Condition~(C4) gives the same read set
\(\mathcal R_{t,g}\), charge set \(\Gamma_{t,g}\), and local privacy
parameter \(\varepsilon_{t,g}\) under \(M\) and \(M'\).

By Condition~(C2),
\[
\Gamma_{t,g}\supseteq
\Gamma^\star(\mathcal R_{t,g}).
\]
Hence, \(a\notin\Gamma_{t,g}\) implies
\[
\mathcal R_{t,g}\cap I_a=\emptyset.
\]
Because \(M\sim_a M'\), the two stores agree at every position
outside \(I_a\). Since
\(\mathcal R_{t,g}\cap I_a=\emptyset\), every position in
\(\mathcal R_{t,g}\) lies outside \(I_a\). Therefore,
\(m_i=m'_i\) for every \(i\in\mathcal R_{t,g}\), and hence
\[
M_{\mathcal R_{t,g}}
=
M'_{\mathcal R_{t,g}}.
\]
Because the scorer uses only \(q_t\) and the read-set memories,
Eq.~\eqref{eq:readset-utility} therefore gives
\[
\mathbf u_{t,g}(M)
=
\mathbf u_{t,g}(M').
\]

Condition~(C1) gives the same view vocabulary
\(\mathcal V_g\), and the fixed sensitivity rule gives the same
\(\Delta u_{t,g}\). By Condition~(C3), the admitted charged
selection uses the EM in Eq.~\eqref{eq:exp-mech}. Since the view vocabulary, score vector, \(\varepsilon_{t,g}\), and
\(\Delta u_{t,g}\) are the same under both stores, the EM gives the
same probability to every \(v\in\mathcal V_g\). Thus,
\[
\Pr[v_{t,g}=v\mid h,M]
=
\Pr[v_{t,g}=v\mid h,M'].
\]
Therefore, the conditional likelihood ratio is one, and the selection
adds zero privacy loss for attribute \(a\).
\end{proof}

\subsection{Proof of Theorem~\ref{thm:transcript}}
\label{app:proof-transcript}

\begin{proof}
Fix a protected attribute \(a\), adjacent stores \(M\sim_a M'\),
an adaptive query strategy, and any fixed value of its auxiliary
information. We first prove the claim for the augmented transcript
\(\tilde{\tau}_T\).

Fix any possible augmented transcript value
\[
\omega
=
\bigl((q_t,v_t,c_t,y_t)\bigr)_{t=1}^{T}.
\]
We compare the probability of producing this same transcript under
\(M\) and \(M'\):
\[
\Pr[\tilde{\tau}_T=\omega\mid M]
\quad\text{and}\quad
\Pr[\tilde{\tau}_T=\omega\mid M'].
\]
For each slot-level operation \(g_{t,j}\) along this transcript, let
\(h_{t,j}\) denote the operation-level prefix immediately before
processing \(g_{t,j}\), as defined in the main text. Because both
probabilities concern the same transcript \(\omega\), they use the
same prefix \(h_{t,j}\) at the corresponding operation.

Initially, all ledgers are zero. Suppose that the ledger values are
the same immediately before a slot-level operation. Under the same
\(h_{t,j}\), Condition~(C4) gives the same read set, charge set, and
local privacy parameter under \(M\) and \(M'\). Because the ledger
values and caps are also the same, Condition~(C5) gives the same
cap-check result. If the selection is admitted, both executions add
the same \(\varepsilon_{t,g}\) to the same ledgers; otherwise, neither
execution changes the ledgers. Thus, the ledger values are also the
same immediately after the operation. By induction, they are the same
immediately before every slot-level operation along \(\omega\).

We now compare the view released at one slot-level operation
\(g=g_{t,j}\), conditional on the same prefix \(h_{t,j}\). There are
three cases.

\begin{itemize}

\item {Case 1: The released view has the same conditional
distribution under \(M\) and \(M'\).}

Condition~(C3) directly gives, for every \(v\in\mathcal V_g\),
\[
\Pr[v_{t,g}=v\mid h_{t,j},M]
=
\Pr[v_{t,g}=v\mid h_{t,j},M'].
\]
This case includes a failed cap check: by the preceding ledger
argument and Condition~(C5), both executions return the same fixed
generic view \(v_g^{\mathrm{gen}}\). Thus, this release adds zero
privacy loss for attribute \(a\).

\item {Case 2: The EM is admitted and
\(a\notin\Gamma_{t,g}\).}

Lemma~\ref{lem:readset}, which relies on Condition~(C2), gives, for every \(v\in\mathcal V_g\),
\[
\Pr[v_{t,g}=v\mid h_{t,j},M]
=
\Pr[v_{t,g}=v\mid h_{t,j},M'].
\]
Thus, this release also adds zero privacy loss for attribute \(a\).

\item {Case 3: The EM is admitted and
\(a\in\Gamma_{t,g}\).}

Write
\[
u_M(v)
=
u_{t,g}(q_t,v;M_{\mathcal R_{t,g}}),
\qquad
u_{M'}(v)
=
u_{t,g}(q_t,v;M'_{\mathcal R_{t,g}}).
\]
Condition~(C1) gives the same view vocabulary, Condition~(C4) gives
the same read set and \(\varepsilon_{t,g}\), and both executions use
the same policy-fixed positive bound \(\Delta u_{t,g}\). Because
\(a\in\Gamma_{t,g}\) and \(M\sim_a M'\),
Eq.~\eqref{eq:score-sensitivity} gives
\[
|u_M(v)-u_{M'}(v)|
\le
\Delta u_{t,g}
\qquad
\text{for every }v\in\mathcal V_g.
\]

Let
\[
Z_M
=
\sum_{v'\in\mathcal V_g}
\exp\!\left(
\frac{\varepsilon_{t,g}u_M(v')}
     {2\Delta u_{t,g}}
\right),
\qquad
Z_{M'}
=
\sum_{v'\in\mathcal V_g}
\exp\!\left(
\frac{\varepsilon_{t,g}u_{M'}(v')}
     {2\Delta u_{t,g}}
\right).
\]
The sensitivity bound gives
\[
u_M(v)-u_{M'}(v)\le\Delta u_{t,g}
\quad\Longrightarrow\quad
\exp\!\left(
\frac{\varepsilon_{t,g}
      (u_M(v)-u_{M'}(v))}
     {2\Delta u_{t,g}}
\right)
\le
e^{\varepsilon_{t,g}/2}.
\]
It also gives, for every \(v'\in\mathcal V_g\),
\[
u_{M'}(v')\le u_M(v')+\Delta u_{t,g}
\quad\Longrightarrow\quad
Z_{M'}
=
\sum_{v'\in\mathcal V_g}
\exp\!\left(
\frac{\varepsilon_{t,g}u_{M'}(v')}
     {2\Delta u_{t,g}}
\right)
\le
\sum_{v'\in\mathcal V_g}
\exp\!\left(
\frac{\varepsilon_{t,g}
      (u_M(v')+\Delta u_{t,g})}
     {2\Delta u_{t,g}}
\right)
=
e^{\varepsilon_{t,g}/2}Z_M.
\]

The EM assigns positive probability to every view in the common
vocabulary. Therefore, Eq.~\eqref{eq:exp-mech} gives
\[
\begin{aligned}
\frac{\Pr[v_{t,g}=v\mid h_{t,j},M]}
     {\Pr[v_{t,g}=v\mid h_{t,j},M']}
&=
\exp\!\left(
\frac{\varepsilon_{t,g}
      (u_M(v)-u_{M'}(v))}
     {2\Delta u_{t,g}}
\right)
\frac{Z_{M'}}{Z_M}
\\
&\le
e^{\varepsilon_{t,g}/2}
e^{\varepsilon_{t,g}/2}
=
e^{\varepsilon_{t,g}}.
\end{aligned}
\]
Thus, this release adds at most
\(\varepsilon_{t,g}\) privacy loss for attribute \(a\).

\end{itemize}

The remaining transcript components add no privacy loss. Because the
same transcript \(\omega\) fixes the same released history, the
adaptive query strategy gives
\[
\Pr[q_t\mid\tilde{\tau}_{<t},M]
=
\Pr[q_t\mid\tilde{\tau}_{<t},M'].
\]
The preceding induction also gives the same ledger values under the
two stores. Condition~(C6) therefore gives the same conditional
distribution for the observable control:
\[
\Pr[c_t\mid
q_t,v_t,\tilde{\tau}_{<t},\{L_b\}_{b\in\mathcal A},M]
=
\Pr[c_t\mid
q_t,v_t,\tilde{\tau}_{<t},\{L_b\}_{b\in\mathcal A},M'].
\]
Equation~\eqref{eq:fixed-llm} likewise gives
\[
\Pr[y_t\mid q_t,v_t,\tau_{<t},M]
=
\Pr[y_t\mid q_t,v_t,\tau_{<t},M'].
\]
Hence, the conditional likelihood ratios of \(q_t\), \(c_t\), and
\(y_t\) are all one. Therefore, in the likelihood ratio for the fixed transcript
\(\omega\), the factors corresponding to \(q_t\), \(c_t\), and
\(y_t\) cancel, and only the slot-level view probabilities remain.
Among them, only admitted EM selections with
\(a\in\Gamma_{t,g}\) can contribute a factor greater than one.

Applying the chain rule in chronological order and multiplying the
bounds above gives
\[
\begin{aligned}
\Pr[\tilde{\tau}_T=\omega\mid M]
&\le
\left(
\prod_{\substack{(t,g)\text{ with an admitted EM}\\
a\in\Gamma_{t,g}}}
e^{\varepsilon_{t,g}}
\right)
\Pr[\tilde{\tau}_T=\omega\mid M']
=
\exp\!\left(
\sum_{\substack{(t,g)\text{ with an admitted EM}\\
a\in\Gamma_{t,g}}}
\varepsilon_{t,g}
\right)
\Pr[\tilde{\tau}_T=\omega\mid M'].
\end{aligned}
\]

Let \(L_a(\omega)\) be the final value of the ledger for attribute
\(a\) along the execution path producing \(\omega\). Every term in
the preceding sum is added to \(L_a\). Since the ledger starts at zero
and never decreases,
\[
\sum_{\substack{(t,g)\text{ with an admitted EM}\\
a\in\Gamma_{t,g}}}
\varepsilon_{t,g}
\le
L_a(\omega).
\]
Condition~(C5) maintains \(L_a(\omega)\le B_a\). Therefore,
\[
\Pr[\tilde{\tau}_T=\omega\mid M]
\le
e^{L_a(\omega)}
\Pr[\tilde{\tau}_T=\omega\mid M']
\le
e^{B_a}
\Pr[\tilde{\tau}_T=\omega\mid M'].
\]

Now let \(\mathcal O\) be any set of possible augmented transcripts.
Summing the preceding inequality over all \(\omega\in\mathcal O\)
gives
\[
\begin{aligned}
\Pr[\tilde{\tau}_T\in\mathcal O\mid M]
&=
\sum_{\omega\in\mathcal O}
\Pr[\tilde{\tau}_T=\omega\mid M]
\le
e^{B_a}
\sum_{\omega\in\mathcal O}
\Pr[\tilde{\tau}_T=\omega\mid M']
=
e^{B_a}
\Pr[\tilde{\tau}_T\in\mathcal O\mid M'].
\end{aligned}
\]
Thus, the augmented transcript \(\tilde{\tau}_T\) satisfies
Eq.~\eqref{eq:transcript-dp} for attribute \(a\).

Finally, the ordinary transcript \(\tau_T\) is obtained from
\(\tilde{\tau}_T\) by removing the selected views. Therefore,
\(\tau_T\) satisfies the same \(B_a\)-DP bound by post-processing.
\end{proof}

\subsection{Proof of Corollary~\ref{cor:attribute-cover}}
\label{app:attribute-cover}

\begin{proof}
Fix an adaptive query strategy and a transcript event
\(\mathcal O\). If \(D=\emptyset\), then \(M=M'\). Moreover, the
empty set is a valid cover of \(D\) with cost zero, so
\(\kappa(M,M')=0\), and the claimed inequality holds with equality.

Assume \(D\neq\emptyset\). Since \(\kappa(M,M')<\infty\), at least
one protected-group cover of \(D\) exists. Let \(\mathcal C\) be a
cover attaining the minimum in Eq.~\eqref{eq:attribute-cover}, and
list its elements as
\[
\mathcal C=\{a_1,\ldots,a_k\}.
\]
Then
\[
D\subseteq\bigcup_{r=1}^{k}I_{a_r},
\qquad
\sum_{r=1}^{k}B_{a_r}
=
\kappa(M,M').
\]

Set \(U_0=\emptyset\), and for \(r=1,\ldots,k\), define
\[
U_r=\bigcup_{s=1}^{r}I_{a_s}.
\]
For \(r=0,\ldots,k\), let
\(M^{(r)}=(m_i^{(r)})_{i=1}^{n}\), where
\[
m_i^{(r)}
=
\begin{cases}
m'_i, & i\in U_r,\\
m_i,  & i\notin U_r.
\end{cases}
\]
Because \(U_0=\emptyset\), we have \(M^{(0)}=M\).
At every position \(i\), \(M^{(r)}\) uses either \(m_i\) or \(m'_i\).
Since \(M\) and \(M'\) have the same scaffold and both values are
slot-compatible at position \(i\), each \(M^{(r)}\) is a valid
same-scaffold store.

For \(r=1,\ldots,k\), we have
\[
U_r=U_{r-1}\cup I_{a_r}.
\]
Positions in \(U_{r-1}\) have value \(m'_i\) in both
\(M^{(r-1)}\) and \(M^{(r)}\), while positions outside \(U_r\) have
value \(m_i\) in both stores. Hence, the two stores can differ only at
positions in
\[
U_r\setminus U_{r-1}
=
I_{a_r}\setminus U_{r-1}
\subseteq I_{a_r}.
\]
Together with their common scaffold and slot-compatible contents, this
gives
\[
M^{(r-1)}\sim_{a_r}M^{(r)}.
\]

Because \(\mathcal C\) covers \(D\), we have \(D\subseteq U_k\).
If \(i\in U_k\), then \(m_i^{(k)}=m'_i\) by construction. If
\(i\notin U_k\), then \(i\notin D\), so \(m_i=m'_i\), and therefore
\[
m_i^{(k)}=m_i=m'_i.
\]
Thus,
\[
M^{(k)}=M'.
\]

Applying Theorem~\ref{thm:transcript} to each adjacent pair gives
\[
\Pr[\tilde{\tau}_T\in\mathcal O\mid M^{(r-1)}]
\le
e^{B_{a_r}}
\Pr[\tilde{\tau}_T\in\mathcal O\mid M^{(r)}]
\]
for \(r=1,\ldots,k\). Chaining these inequalities gives
\[
\begin{aligned}
\Pr[\tilde{\tau}_T\in\mathcal O\mid M^{(0)}]
&\le
\left(\prod_{r=1}^{k}e^{B_{a_r}}\right)
\Pr[\tilde{\tau}_T\in\mathcal O\mid M^{(k)}]
=
\exp\!\left(\sum_{r=1}^{k}B_{a_r}\right)
\Pr[\tilde{\tau}_T\in\mathcal O\mid M^{(k)}].
\end{aligned}
\]
Using \(M^{(0)}=M\), \(M^{(k)}=M'\), and
\(\sum_{r=1}^{k}B_{a_r}=\kappa(M,M')\), we obtain
\[
\Pr[\tilde{\tau}_T\in\mathcal O\mid M]
\le
e^{\kappa(M,M')}
\Pr[\tilde{\tau}_T\in\mathcal O\mid M'].
\]
This proves Eq.~\eqref{eq:composite-transcript}. The same bound holds
for the ordinary transcript because it is obtained from
\(\tilde{\tau}_T\) by removing the selected views.
\end{proof}

\subsection{Posterior-Odds Interpretation}
\label{app:posterior-odds}

Let \(H\in\{0,1\}\) indicate which adjacent store is used, with
\(H=1\) corresponding to \(M\) and \(H=0\) corresponding to \(M'\).
Assume \(0<\Pr[H=1]<1\). For any augmented-transcript event
\(\mathcal O\) having positive probability under both stores, Bayes'
rule gives
\[
\Pr[H=1\mid\tilde{\tau}_T\in\mathcal O]
=
\frac{
\Pr[\tilde{\tau}_T\in\mathcal O\mid M]\Pr[H=1]
}{
\Pr[\tilde{\tau}_T\in\mathcal O]
},
\qquad
\Pr[H=0\mid\tilde{\tau}_T\in\mathcal O]
=
\frac{
\Pr[\tilde{\tau}_T\in\mathcal O\mid M']\Pr[H=0]
}{
\Pr[\tilde{\tau}_T\in\mathcal O]
}.
\]
Taking the ratio of these two expressions gives
\[
\begin{aligned}
\frac{
\Pr[H=1\mid\tilde{\tau}_T\in\mathcal O]
}{
\Pr[H=0\mid\tilde{\tau}_T\in\mathcal O]
}
&=
\frac{
\Pr[\tilde{\tau}_T\in\mathcal O\mid M]\Pr[H=1]
}{
\Pr[\tilde{\tau}_T\in\mathcal O\mid M']\Pr[H=0]
}.
\end{aligned}
\]
Dividing both sides by the prior odds
\(\Pr[H=1]/\Pr[H=0]\) yields
\[
\frac{
  \Pr[H=1\mid\tilde{\tau}_T\in\mathcal O]/
  \Pr[H=0\mid\tilde{\tau}_T\in\mathcal O]
}{
  \Pr[H=1]/\Pr[H=0]
}
=
\frac{
  \Pr[\tilde{\tau}_T\in\mathcal O\mid M]
}{
  \Pr[\tilde{\tau}_T\in\mathcal O\mid M']
}.
\]

Theorem~\ref{thm:transcript} gives the upper bound
\(e^{B_a}\). Since \(\sim_a\) is symmetric, applying the theorem with
\(M\) and \(M'\) exchanged gives the lower bound \(e^{-B_a}\).
Therefore,
\[
e^{-B_a}
\le
\frac{
  \Pr[H=1\mid\tilde{\tau}_T\in\mathcal O]/
  \Pr[H=0\mid\tilde{\tau}_T\in\mathcal O]
}{
  \Pr[H=1]/\Pr[H=0]
}
\le
e^{B_a}.
\]
Thus, observing any transcript event can change the prior odds by at
most a factor \(e^{B_a}\) in either direction. The same statement
holds for the ordinary transcript.

\subsection{Failure Modes under Omitted Interface Invariants}
\label{app:contract-failures}

\begin{proposition}[Failure modes under omitted interface invariants]
\label{prop:contract-failures}
Within the \textsc{DP-MemView} architecture and without additional
safeguards, each of the following four categories is necessary:
common view support, complete attribute charging, exclusion or
privatization of raw-memory-dependent observable releases, and
pathwise budget enforcement. For each category, there exists an
interface that satisfies the other three categories but violates
\(B_a\)-DP.
\end{proposition}

\begin{proof}
We give one counterexample for each category. All policies and
interface components not mentioned in a construction remain fixed and
store-independent.

\paragraph{Common view support (C1).}
Consider a fixed admitted operation for which the read set, charge
set, local privacy parameter, cap decision, and all downstream
channels satisfy the remaining contract categories. Let only the view
vocabulary depend on the store. Suppose that
\[
v_\star\in\mathcal V_g(M),
\qquad
v_\star\notin\mathcal V_g(M')
\]
for adjacent stores \(M\sim_a M'\), and run the EM over the respective
vocabularies. Since the EM assigns positive probability to every
candidate view, the transcript event
\[
\mathcal O
=
\{\tilde{\tau}_T:v_{t,g}=v_\star\}
\]
satisfies
\[
\Pr[\tilde{\tau}_T\in\mathcal O\mid M]
=
\gamma>0,
\qquad
\Pr[\tilde{\tau}_T\in\mathcal O\mid M']
=
0.
\]
Any finite \(B_a\) would require
\[
\gamma
\le
e^{B_a}
\Pr[\tilde{\tau}_T\in\mathcal O\mid M']
=
0,
\]
which is impossible. Thus, store-dependent view support precludes any
finite pure-DP guarantee.

\paragraph{Complete attribute charging (C2).}

Keep the same view vocabulary under both stores, use the valid EM in
Eq.~\eqref{eq:exp-mech}, keep the routing and downstream outputs
store-independent, and retain the pathwise cap check. We change only
the charge set.

Choose two overlapping protected groups \(I_a\) and \(I_b\), where
\(b\neq a\), and a position
\(i^\star\in\mathcal R_{t,g}\cap I_a\cap I_b\). Choose adjacent
stores \(M\sim_a M'\) that differ only at \(i^\star\). Since the read
set touches both groups,
\[
\{a,b\}\subseteq\Gamma^\star(\mathcal R_{t,g}).
\]
Nevertheless, declare
\[
\Gamma_{t,g}=\{b\},
\]
so the operation checks and charges only \(b\), while omitting the
affected attribute \(a\).

Fix a local parameter \(\varepsilon_0>0\). Choose an integer
\(N>2B_a/\varepsilon_0\), and choose
\(B_b\ge N\varepsilon_0\). For each \(n=1,\ldots,N\), suppose that
the first \(n-1\) repetitions have been admitted. Then
\(L_b=(n-1)\varepsilon_0\), and the \(n\)-th cap check gives
\[
L_b+\varepsilon_0
=
n\varepsilon_0
\le
N\varepsilon_0
\le
B_b.
\]
Thus, by induction, all \(N\) repetitions are admitted. After the final
repetition,
\[
L_b=N\varepsilon_0\le B_b,
\qquad
L_a=0.
\]

Let \(\mathcal V_g=\{v_0,v_1\}\). Write a generic store as
\(X=(x_i)_{i=1}^{n}\), and choose a binary function \(\phi\) such that
\(\phi(m_{i^\star})=1\) and \(\phi(m'_{i^\star})=0\). Define
\[
u(v_1;X)=\phi(x_{i^\star}),
\qquad
u(v_0;X)=1-\phi(x_{i^\star}).
\]
The resulting scores are
\[
u(v_1;M)=1,\quad u(v_0;M)=0,
\qquad
u(v_1;M')=0,\quad u(v_0;M')=1.
\]
Since all scores lie in \([0,1]\), the bound
\(\Delta u_{t,g}=1\) is valid.

Suppressing the common operation prefix,
Eq.~\eqref{eq:exp-mech} gives
\[
\Pr[v_{t,g}=v_1\mid M]
=
\frac{e^{\varepsilon_0/2}}
     {1+e^{\varepsilon_0/2}},
\qquad
\Pr[v_{t,g}=v_1\mid M']
=
\frac{1}
     {1+e^{\varepsilon_0/2}}.
\]
Hence,
\[
\frac{\Pr[v_{t,g}=v_1\mid M]}
     {\Pr[v_{t,g}=v_1\mid M']}
=
e^{\varepsilon_0/2}.
\]

Let \(\mathcal O_N\) be the augmented-transcript event that all
\(N\) repetitions release \(v_1\). Keep the read set, scores, and
local privacy parameter fixed across repetitions, and use fresh EM
randomness at each repetition. All other transcript components are
the same under \(M\) and \(M'\). Therefore,
\[
\begin{aligned}
\frac{
\Pr[\tilde{\tau}_T\in\mathcal O_N\mid M]
}{
\Pr[\tilde{\tau}_T\in\mathcal O_N\mid M']
}
&=
\left(
\frac{\Pr[v_{t,g}=v_1\mid M]}
     {\Pr[v_{t,g}=v_1\mid M']}
\right)^N
\\
&=
e^{N\varepsilon_0/2}.
\end{aligned}
\]
Because \(N>2B_a/\varepsilon_0\),
\[
e^{N\varepsilon_0/2}>e^{B_a}.
\]
Thus, the transcript violates \(B_a\)-DP even though \(L_a\) remains
zero throughout the execution. Omitting the affected attribute \(a\)
from \(\Gamma_{t,g}\) therefore invalidates the claimed
\(B_a\)-DP guarantee.

\paragraph{Raw-memory-dependent observable releases
(C3--C4/C6).}
First, keep the common vocabulary
\(\mathcal V_g=\{v_0,v_1\}\), but replace the EM with deterministic
argmax selection. Use the reversed scores from the preceding
construction. Then
\[
\Pr[v_{t,g}=v_1\mid M]=1,
\qquad
\Pr[v_{t,g}=v_1\mid M']=0.
\]
Hence, no finite \(B_a\) can satisfy DP. This is the failure used by
the CappedArgmax ablation.

A raw-memory-dependent control signal causes the same problem. Choose
an adjacent pair in which \(M_{I_a}\) contains evidence for \(a\) and
\(M'_{I_a}\) is \(a\)-neutral, and release
\[
c
=
\mathbf 1\{M_{I_a}\text{ contains evidence for }a\}.
\]
Then
\[
\Pr[c=1\mid M]=1,
\qquad
\Pr[c=1\mid M']=0.
\]
The same one-sided event can be created by a memory-dependent slot
activation, fallback decision, ledger-status signal, or released
score. Therefore, every raw-memory-dependent observable output must
either have the same distribution under adjacent stores or be
released by a valid private mechanism.

\paragraph{Pathwise cumulative cap (C5).}
Keep the common vocabulary, complete charging, and the EM, but remove
the cumulative cap. For every operation, set
\[
\Gamma_{t,g}=\{a\}.
\]
Use the same two-view scores as in the C2 construction, with local
parameter \(\varepsilon_0\), and repeat the operation \(N\) times.
Each release is correctly charged to \(a\), but no cap prevents the
repetitions.

For the transcript event \(\mathcal O_N\) that every operation
releases \(v_1\),
\[
\frac{
\Pr[\tilde{\tau}_T\in\mathcal O_N\mid M]
}{
\Pr[\tilde{\tau}_T\in\mathcal O_N\mid M']
}
=
e^{N\varepsilon_0/2}.
\]
Choosing
\[
N>\frac{2B_a}{\varepsilon_0}
\]
makes this ratio larger than \(e^{B_a}\). Thus, valid and correctly
charged local EM releases do not provide a fixed \(B_a\)-DP guarantee
unless their cumulative cost is capped on every execution path.
\end{proof}

\section{Supplementary Experimental Details}
\label{app:experiments}

\subsection{Evaluation Data Construction}
\label{app:data-construction}

We use two paired evaluation tracks. PairedMem provides exact
\(M/M'\) control, while the corpus-sourced track tests whether the
results persist with memory text adapted from public corpora.
Table~\ref{tab:data-tracks} summarizes the two tracks. Across both tracks, we keep the queries, router, read and charge
policies, scorer, view vocabularies, and gold rules fixed.

\begin{table*}[h]
\centering
\small
\caption{Summary and scale of the two paired evaluation tracks.}
\label{tab:data-tracks}
\setlength{\tabcolsep}{4.8pt}
\renewcommand{\arraystretch}{1.4}
\begin{tabular}{@{}
p{0.13\textwidth}
p{0.22\textwidth}
p{0.37\textwidth}
p{0.23\textwidth}
@{}}
\toprule
\textbf{Track}
&
\textbf{Memory text}
&
\textbf{\(M/M'\) construction}
&
\textbf{Scale}
\\
\midrule

\textbf{PairedMem}
&
Synthetic memory text generated from fixed specifications
&
Same 28-position scaffold; \(M\) and \(M'\) differ only within
\(I_a\); 25\% of pairs contain overlapping protected groups
&
320 \(M/M'\) pairs \(\times\) 32 queries
\newline
\(=10{,}240\) paired query turns
\\

\textbf{Corpus-sourced transfer}
&
Public dialogue histories containing first-person disclosures
&
Same PairedMem scaffold and shared background; \(M'\) removes only
the target signal
&
80 \(M/M'\) pairs \(\times\) 32 queries
\newline
\(=2{,}560\) paired query turns
\\
\bottomrule
\end{tabular}
\end{table*}

\subsection{PairedMem: Controlled Synthetic Track}
\label{app:pairedmem}

PairedMem contains 320 adjacent-store pairs generated from fixed
specifications. Table~\ref{tab:pairedmem-pipeline} summarizes the four
construction steps.

\begin{table*}[h]
\centering
\small
\caption{PairedMem construction.}
\label{tab:pairedmem-pipeline}
\setlength{\tabcolsep}{5pt}
\renewcommand{\arraystretch}{1.15}
\begin{tabular}{@{}
p{0.12\textwidth}
p{0.47\textwidth}
p{0.37\textwidth}
@{}}
\toprule
\textbf{Step}
&
\textbf{What is done}
&
\textbf{Generation / validation}
\\
\midrule

\textbf{S1}
\newline Pair design
&
Define paired-store specifications with the same 28-position
scaffold. Each pair may differ only within one protected group
\(I_a\).
&
Code creates the specifications and pair-level split
\\

\textbf{S2}
\newline Memory writing
&
Convert each specification into natural-language memory sentences
for \(M\) and \(M'\).
&
Mistral-7B writes the sentence wording
\\

\textbf{S3}
\newline Pair validation
&
Verify that \(M\) contains the target evidence, \(M'\) removes it,
the stores agree outside \(I_a\), and non-target evidence is preserved.
&
Claude Sonnet 4.6 checks whether the
text matches the specification and reads naturally
\\

\textbf{S4}
\newline Queries and gold
&
Create one 32-query trajectory and deterministic gold response
behaviors for each pair.
&
Mistral-7B writes the queries, and Claude Sonnet 4.6 checks each
trajectory
\\

\bottomrule
\end{tabular}
\end{table*}

\paragraph{Why paired synthetic stores?}
No public corpus provides two stores with the same scaffold,
non-target memories, and query trajectory while differing only within
one protected group \(I_a\). Comparing unrelated users would also
change writing style and background facts. We therefore generated the
paired stores directly.

\paragraph{S1: Pair design.}
We created 384 pair specifications, with 96 for each of four protected
attributes: type-2 diabetes, pregnancy, severe unsecured debt, and
recent involuntary job loss. Each store has 28 fixed positions in six
slots, and each \(I_a\) contains 4--6 positions. The paired stores
agree outside \(I_a\). In 25\% of the specifications, protected groups
overlap.

\paragraph{S2: Text realization.}
Each specification fixes the memory positions, slots, target evidence,
and non-target clauses that must remain unchanged. Mistral-7B converts
the specification into natural-language memory sentences. All
resulting text is synthetic and contains no real personal information.

\paragraph{S3: Validation and selection.}
Code checks that \(M\) and \(M'\) have the same scaffold, differ only
within \(I_a\), contain the required target evidence on the correct
side, and preserve non-target evidence at overlapping positions.
Claude Sonnet 4.6 separately checks whether each sentence matches its
specification and reads naturally. Of the 384 candidate pairs, 381
passed validation. We selected 320 using fixed attribute, template,
and overlap quotas and retained the remaining 61 as reserves.

\paragraph{S4: Queries and deterministic gold.}
Each selected pair has a 32-query trajectory consisting of two
disjoint 16-query blocks. Each block contains eight target-relevant,
four same-domain neutral, and four distractor queries. The query
generator receives the query type, slot, and public intent, but not
\(M\) or \(M'\). Code determines the read sets, charge sets,
applicable view vocabularies, and gold response constraints from the
fixed policy. The same trajectory is used for \(M\), \(M'\), and
every compared interface.

\subsection{Controlled Corpus-Sourced Transfer}
\label{app:corpus-transfer}

The transfer track contains 80 adjacent-store pairs built from public
corpus text. We keep the PairedMem scaffold, paired difference,
queries, and gold rules fixed and replace only the memory sentences.
Table~\ref{tab:corpus-transfer-pipeline} summarizes the four
construction steps.

\begin{table*}[h]
\centering
\small
\caption{Controlled corpus-sourced transfer construction.}
\label{tab:corpus-transfer-pipeline}
\setlength{\tabcolsep}{4.5pt}
\renewcommand{\arraystretch}{1.4}
\begin{tabular}{@{}
p{0.15\textwidth}
p{0.43\textwidth}
p{0.39\textwidth}
@{}}
\toprule
\textbf{Step}
&
\textbf{What is done}
&
\textbf{Generation / validation}
\\
\midrule

\textbf{T1}
\newline Pair skeleton
&
Copy a matching PairedMem scaffold for each history--attribute
combination. Keep the positions, slots, protected groups, and target
positions fixed.
&
Code creates 80 pair specifications and assigns the same background
to \(M\) and \(M'\)
\\

\textbf{T2}
\newline Memory sourcing
&
Replace the synthetic memory sentences with corpus-sourced background
and target evidence. Create \(M'\) by removing only the target signal.
&
Memory text is drawn from public corpora; \(M'\) is formed by minimally
editing the corresponding evidence sentences
\\

\textbf{T3}
\newline Pair validation
&
Verify the shared scaffold, equality outside \(I_a\), target removal,
slot compatibility, and preservation of non-target evidence.
&
Each sentence and its edited counterpart are manually screened; code
checks the pair structure and gold consistency
\\

\textbf{T4}
\newline Frozen evaluation
&
Reuse the PairedMem queries, router, read and charge policies, scorer,
view vocabularies, and gold rules.
&
Code applies the fixed protocol
\\

\bottomrule
\end{tabular}
\end{table*}

\paragraph{T1: Pair skeleton.}
We form 80 pair specifications by combining 20 background histories
with four protected attributes. Each specification copies the
28-position, six-slot scaffold and target positions from a matching
PairedMem pair. The background is identical in \(M\) and \(M'\), and
the two stores differ only within \(I_a\).

\paragraph{T2: Memory sourcing.}
Memory text comes from three public datasets. LoCoMo persona and event
histories provide the background text
\cite{maharana2024locomo}. The \texttt{all-processed} configuration
of the Medical Question Answering Datasets collection provides
first-person evidence for diabetes and pregnancy, as well as some debt
and job-loss evidence \cite{malikehmedicalqa}. EmpatheticDialogues
provides additional debt and job-loss evidence
\cite{rashkin2019towards}.

For each evidence sentence in \(M\), the corresponding sentence in
\(M'\) is a minimal edit that removes only the target-attribute signal
while preserving the topic, sentence structure, and any non-target
evidence. Thus, the paired sentences retain similar wording while
differing in the target signal.

\paragraph{T3: Pair validation.}
We retain only first-person statements about the speaker's current
state. We exclude resolved past states, hypothetical statements,
metaphors, and statements about other people. Each pair is checked for
the shared scaffold, equality outside \(I_a\), target removal, slot
compatibility, and preservation of non-target evidence. All 80 pairs
pass these checks.

\paragraph{T4: Frozen evaluation.}
Each pair reuses the 32-query trajectory of its matching PairedMem
pair, yielding 2,560 paired query turns. The router, read and charge
policies, scorer, view vocabularies, response models, and deterministic
gold rules remain unchanged. 

\subsection{Compared Interface Implementations}
\label{app:compared-interfaces}

Table~\ref{tab:interface-implementations} summarizes the input passed
to the response LLM, how each method constructs that input, and whether
it uses a DP mechanism. All methods use the same queries, response
model, router, and output limit.

\begin{table*}[h]
\centering
\scriptsize
\caption{Implementation of the compared interfaces and whether each
uses a DP mechanism. \(K\) is the online budget granularity, and
\(C_a\) is the planned number of charged selections for attribute
\(a\).}
\label{tab:interface-implementations}
\setlength{\tabcolsep}{4.5pt}
\renewcommand{\arraystretch}{1.12}
\begin{tabular}{@{}
p{0.16\textwidth}
p{0.15\textwidth}
p{0.53\textwidth}
p{0.08\textwidth}
@{}}
\toprule
\textbf{Method}
&
\textbf{Response-LLM input}
&
\textbf{How the input is formed}
&
\textbf{Uses DP}
\\
\midrule

\multicolumn{4}{@{}l}{\textbf{Memory-independent}}\\[-1mm]

GenericOnly
&
Fixed generic view
&
Uses the same fixed slot-level generic string under \(M\) and \(M'\).
&
\multicolumn{1}{c}{No}
\\

\addlinespace[1pt]
\multicolumn{4}{@{}l}{\textbf{Memory text}}\\[-1mm]

RawReadSet
&
Raw read-set text
&
Passes the memory items at \(\mathcal R_{t,g}\), ordered by position.
&
\multicolumn{1}{c}{No}
\\

TypedMask
\newline
\emph{(adaptation of MemPrivacy)}
&
Item-masked read set
&
Replaces protected read items with slot-typed placeholders and passes
the remaining items unchanged.
&
\multicolumn{1}{c}{No}
\\

TaskMin
\newline
\emph{(adaptation of AirGapAgent)}
&
Fixed task-family subset
&
Passes the fixed union of positions used by target-relevant and
same-domain query families; distractor-only positions are excluded.
&
\multicolumn{1}{c}{No}
\\

OutputFilter
&
Raw read-set text during generation
&
Generates from the RawReadSet input and then deletes explicit protected
labels and sensitive-state phrases.
&
\multicolumn{1}{c}{No}
\\

\addlinespace[1pt]
\multicolumn{4}{@{}l}{\textbf{Proposed modes}}\\[-1mm]

\textbf{DP-MemView (on)}
&
EM-sampled public view or generic fallback
&
Uses the shared view vocabulary and scorer, with
\(\varepsilon_{t,g}
=\min_{a\in\Gamma_{t,g}}B_a/K\), complete charging, ledgers, and a
prospective cap check.
&
\multicolumn{1}{c}{Yes}
\\

\textbf{DP-MemView (pre)}
&
EM-sampled public view or generic fallback
&
Uses the same view vocabulary, scorer, charging, ledgers, and cap
check, with
\(\varepsilon_{t,g}
=\min_{a\in\Gamma_{t,g}}B_a/C_a\).
&
\multicolumn{1}{c}{Yes}
\\

\bottomrule
\end{tabular}
\end{table*}

\subsection{Evaluation Protocol}
\label{app:evaluation-protocol}

\paragraph{Models and execution environment.}
Table~\ref{tab:eval-models} lists the models used for response
generation and evaluation. For each response model, all compared
interfaces use the same decoding configuration.

\begin{table*}[h]
\centering
\scriptsize
\setlength{\tabcolsep}{3pt}
\renewcommand{\arraystretch}{1.05}
\caption{Models used for response generation and evaluation.}
\label{tab:eval-models}
\begin{tabular}{@{}
p{0.23\textwidth}
p{0.70\textwidth}
@{}}
\toprule
\textbf{Role} & \textbf{Model} \\
\midrule

Response generation
&
Qwen2.5-7B-Instruct, Llama-3.1-8B-Instruct, and Gemma-2-9B-it
\\

Embedding auditor
&
Frozen \texttt{sentence-transformers/all-mpnet-base-v2} encoder
with calibrated logistic regression
\\

Pairwise privacy auditor
&
Mistral-Small-24B-Instruct-2501 
\\

Utility judge
&
Claude Sonnet 4.6
\\

\bottomrule
\end{tabular}
\end{table*}

Response generation, embedding auditing, and the pairwise LLM audit
were run in Google Colab using an NVIDIA A100 GPU. Claude Sonnet 4.6
was accessed through its API for utility judging.

\paragraph{Embedding privacy auditing.}
Our primary privacy auditor uses frozen response embeddings. For each
response model and transcript prefix, we embed each assistant response,
average the turn embeddings, and \(L_2\)-normalize the resulting prefix
vector. For each protected attribute and prefix length, we train a
separate logistic-regression auditor to distinguish transcripts from
\(M\) and \(M'\). The auditor is fitted on PairedMem training pairs,
sigmoid-calibrated on validation pairs, and evaluated on held-out test
pairs. The same fitted auditor is applied to the corpus-sourced transfer
track without refitting or recalibration. We report macro-AUC across
the four protected attributes and TPR at 5\% FPR. 

\paragraph{Utility judging.}
At \(T=16\), the blinded judge receives the query--response transcript
and a rubric generated from the fixed gold constraints for that store
side. The rubric is fixed before response evaluation and is shared
across all compared interfaces for the same pair and store side. The
judge does not receive the method name or the \(M/M'\) label.

The judge assigns 0--4 scores for relevance, correctness,
actionability, and personalization. It also scores each applicable
gold behavior on the same 0--4 scale. All scores are normalized by
dividing by 4. The mean of the four normalized holistic scores defines
\(U\).

For \(M\), \(\mathrm{tRec}\) is the mean normalized score over
target-required behaviors. For \(M'\), \(\mathrm{Unsup}\) is the mean
normalized score over target-specific behaviors that are unsupported
by that store. Thus, higher \(\mathrm{tRec}\) indicates better use of
supported target information, whereas lower \(\mathrm{Unsup}\)
indicates less unsupported target-specific personalization. 

\paragraph{Experimental view scorer and sensitivity.}
DP-MemView does not require a particular score function. Any scorer
with values in \([0,1]\) may be used if it has a valid sensitivity
bound satisfying Eq.~\eqref{eq:score-sensitivity}. The experiments use
the following fixed scorer:
\[
u_{t,g}(q_t,v;M_{\mathcal R_{t,g}})
=
0.4\,r(q_t,v)
+
0.4\,c(v;M_{\mathcal R_{t,g}})
+
0.2\,s(v;M_{\mathcal R_{t,g}}).
\]
Here, \(r,c,s\in[0,1]\). The relevance term \(r\) depends only on
the public query and candidate view. The compatibility term \(c\)
measures whether the view matches the read-set memory, and the
specificity term \(s\) measures whether the read-set memory supports
the specificity of the view. The scorer uses only the current
read-set memory and does not access the paired store.

For adjacent stores \(M\sim_a M'\), fix the same query and candidate
view. The relevance term is then identical under both stores. Writing
\(u_M(v)\) and \(u_{M'}(v)\) for the two scores,
\[
\begin{aligned}
|u_M(v)-u_{M'}(v)|
&\le
0.4\,|c_M(v)-c_{M'}(v)|
+
0.2\,|s_M(v)-s_{M'}(v)|
\\
&\le
0.4+0.2
=
0.6.
\end{aligned}
\]
Therefore,
\[
\left\|
\mathbf u_{t,g}(M)-\mathbf u_{t,g}(M')
\right\|_\infty
\le 0.6,
\]
and the experiments set \(\Delta u_{t,g}=0.6\).

\paragraph{Runs and aggregation.}
GenericOnly, RawReadSet, TypedMask, TaskMin, OutputFilter, and
CappedArgmax are deterministic and are run once per response model.
DP-MemView (on), DP-MemView (pre), IncompleteCharge, and NoCap use
EM seeds \(\{0,1,2\}\). For stochastic methods, measurements are
averaged across seeds within each adjacent pair before aggregation
over pairs.

\end{document}